\documentclass[conference]{IEEEtran}
\usepackage{amsmath,graphicx}
\usepackage{color}
\usepackage{graphicx}
\usepackage{epstopdf}
\usepackage{amsmath}
\usepackage{amssymb}
\usepackage{mathrsfs}
\usepackage[bookmarks=false, pdfpagelabels=false]{hyperref}
\usepackage{tikz}
\usepackage{bm}
\usepackage[english]{babel}
\usepackage{cite}
\usepackage{rotfloat}
\usepackage{mathtools}
\usepackage[font=normalsize,labelfont=bf]{caption}
\usepackage{amsmath}
\usepackage{makecell}
\usepackage{multirow}
\usepackage{booktabs}
\usepackage{colortbl}
\usepackage{multirow}% http://ctan.org/pkg/multirow
\usepackage{hhline}% http://ctan.org/pkg/hhline
\usepackage{stfloats}% <-- added
\usepackage{cuted} %phuong
\usepackage{multicol}
\usepackage{bbm}
\usepackage{cases}
\graphicspath{ {Figures/} }
\newsavebox{\foobox}

\definecolor{kugray5}{RGB}{224,224,224}

\usepackage[normalem]{ulem}
\newcommand\rsout{\bgroup\markoverwith
	{\textcolor{red}{\rule[0.5ex]{2pt}{0.8pt}}}\ULon}



\makeatletter
\newcommand{\ALOOP}[1]{\ALC@it\algorithmicloop\ #1%
	\begin{ALC@loop}}
	\newcommand{\ENDALOOP}{\end{ALC@loop}\ALC@it\algorithmicendloop}

\makeatother

\usepackage{etoolbox}
\let\mybibitem\bibitem
\renewcommand{\bibitem}[1]{%
	\ifstrequal{#1}{nature}
	{\color{blue}\mybibitem{#1}}
	{\color{black}\mybibitem{#1}}%
}

\graphicspath{ {Figures/} }

\newtheorem{theorem}{\textbf{Theorem}}

\DeclareCaptionLabelSeparator{periodspace}{.\quad}

\addto\captionsenglish{}
\allowdisplaybreaks
\usepackage{setspace}
\newcommand{\mI}{\textbf{\textbf{I}}}

\newcommand{\vx}{{\mathbf{x}}}
\newcommand{\vy}{{\mathbf{y}}}

\newcommand{\vv}{{\mathbf{v}}}

\newcommand{\vu}{{\mathbf{u}}}
\newcommand{\vz}{{\mathbf{z}}} 
\newcommand{\vh}{{\mathbf{h}}}

\newcommand{\vw}{{\mathbf{w}}}

\def\b0{{\pmb{0}}}

\usepackage{mdframed} % For the framed box

\usepackage{caption} %phuong
\usepackage{subcaption} %phuong
\usepackage[ruled,vlined]{algorithm2e} % for \KwIn, \KwOut, \For, \tcp, etc.

\begin{document}

\title{Deep-Unfolded Accelerated Projected Gradient for Energy-Efficient Cell-Free Massive MIMO}

\author{
    \IEEEauthorblockN{Phuong~Nam~Tran\IEEEauthorrefmark{1}, Nhan~Thanh~Nguyen\IEEEauthorrefmark{1}, Hien~Quoc~Ngo\IEEEauthorrefmark{2}, Markku~Juntti\IEEEauthorrefmark{1},}
    \IEEEauthorblockA{
    \IEEEauthorrefmark{1}Centre for Wireless Communications (CWC), University of Oulu, Finland, \\
    \IEEEauthorrefmark{2}School of Electronics, Electrical Engineering and Computer Science, Queen’s University Belfast, UK
    \\\{phuong.tran, nhan.nguyen, markku.juntti\}@oulu.fi, hien.ngo@qub.ac.uk
    }
    }
\maketitle

\begin{abstract}

This paper investigates  energy efficiency (EE) maximization for the downlink of cell-free massive multiple-input multiple-output  systems under quality-of-service and per-access point power constraints. We first derive closed-form gradient expressions of the objective function with respect to the power allocation coefficients, and then propose an accelerated projected gradient (APG) approach to solve this problem.
To reduce the computational complexity and runtime of APG, we propose a deep-unfolded APG framework that maps iterative APG updates onto a finite number of neural network layers, where parameters such as step sizes and penalty coefficients are learned from data. The proposed approach produces power allocation solutions through a fixed number of gradient-based updates without the need for line search or manual parameter tuning. 
Numerical results show that the method achieves EE performance comparable to the iterative APG approach while requiring significantly lower computational cost, with up to a 30-fold reduction in floating-point operations under the considered system settings.
% Numerical results for a system with 20 APs, each equipped with 4 antennas and serving 16 single-antenna users, show that the proposed unfoled APG method achieves comparable EE to the iterative APG while reducing the number of floating-point operations by up to 30 times

% Conventional model-based approaches typically rely on iterative optimization with high computational complexity.

\end{abstract}

\begin{IEEEkeywords}
Cell-free massive MIMO, energy efficiency, power control,
accelerated proximal gradient, deep unfolding.
\end{IEEEkeywords}

\IEEEpeerreviewmaketitle

\section{Introduction}

Cell-free massive multiple-input multiple-output (CF-mMIMO) has emerged as a promising architecture for future wireless networks \cite{ngo2017total, ngo2024ultradense}. By eliminating cell boundaries, CF-mMIMO improves spectral efficiency (SE) and provides a more uniform service quality compared with conventional cellular architectures. In such systems, energy efficiency (EE) is a key design objective, since dense deployments of access points (APs) lead to substantial network power consumption\cite{jayaweera2024minimizing, ngo2024ultradense}. Therefore, power control plays a central role in improving the EE of CF-mMIMO systems~\cite{mai2022apg}.

% Existing works on power allocation in CF-mMIMO systems have considered objectives such as power consumption minimization, SE maximization, and EE maximization \cite{ngo2017total, mai2022apg, van2020joint, yan2025efficient}. Most existing methods typically rely on sequential convex approximation (SCA), which solves a sequence of convex subproblems. 

Existing works on power allocation in CF-mMIMO systems have considered objectives such as power consumption minimization, SE maximization, and EE maximization \cite{ngo2017total, mai2022apg, van2020joint, yan2025efficient}. Most existing methods typically rely on sequential convex approximation (SCA) methods. Although such methods can achieve near-optimal performance, they require solving convex subproblems at each iteration, resulting in high computational complexity. To reduce this burden, accelerated projected gradient (APG) methods have been proposed \cite{mai2022apg, farooq2021utility, yan2024joint}, which offer lower per-iteration complexity. However, these methods still require multiple iterations with line search or parameter tuning, leading to high and variable computational cost. Additionally, prior work has focused on maximum-ratio transmission (MRT) precoding~\cite{mai2022apg}, while partial zero-forcing (PZF) precoding, which can achieve higher SE \cite{interdonato2020local}, has not been explored.

% Alternatively, scalable strategies have been proposed~\cite{bjornson2020scalable, buzzi2019user}, including heuristic power allocation, AP selection, and problem decomposition. While these approaches improve scalability, they 
% % often rely on simplified allocation rules that 
% may not fully capture the coupling among variables, particularly in EE maximization under QoS constraints.
% \blue{ In addition, the EE maximization problem involves strong coupling among optimization variables through the signal-to-interference-plus-noise ratio (SINR), while quality-of-service (QoS) constraints further restrict the feasible set, making the problem particularly challenging to solve efficiently.}
% In EE maximization under QoS constraints, such coupling is significant, as the power allocation of each user affects the interference experienced by others.

In CF-mMIMO power control, deep learning (DL) has been explored to reduce computational complexity by replacing iterative optimizers with direct mappings from system parameters such as large-scale fading coefficients to power allocation variables~\cite{rajapaksha2021deep, salaun2022gnn, mishra2024graph}. However, these approaches are typically designed as black-box models, and it is difficult to explicitly enforce system constraints such as QoS and per-AP power constraints. Model-based DL, particularly deep unfolding, provides an alternative by embedding the structure of iterative optimization algorithms into trainable networks \cite{gao2023hybrid, xu2023algorithm, yang2025deep, nguyen2023deep, nguyen2024joint, xu2024leakage, lin2024communication}. This approach retains the structure of conventional optimization algorithms while replacing iterative procedures with a finite number of layers, resulting in fixed computational cost and enhanced interpretability. 
Deep unfolding has been applied in CF-mMIMO to tasks such as precoding or beamforming design \cite{ gao2023hybrid, xu2023algorithm, yang2025deep}, leakage-rate maximization \cite{xu2024leakage}, and activity detection \cite{lin2024communication}. However, these works do not consider EE maximization under QoS and per-AP power constraints.

In this paper, we consider EE maximization for the downlink of a CF-mMIMO system. We first develop an APG-based framework for the EE maximization problem under QoS and per-AP power constraints, where local PZF precoding is employed. However, the resulting APG method remains iterative and requires multiple iterations with step-size and parameter tuning, leading to high computational cost. To address this limitation, we propose a deep-unfolded APG approach by mapping the iterative updates of APG method into a finite number of trainable layers. In the proposed model, the step sizes are learned from data, eliminating the need for line search and iterative tuning. The resulting unfolded model computes the power allocation through a fixed number of gradient-based updates, achieving fixed computational complexity and fast convergence, while explicitly enforcing the per-AP power constraints and incorporating QoS requirements through a penalty-based formulation. Numerical results show that the proposed method achieves EE performance comparable to iterative APG while requiring significantly lower computational complexity.

% =============================================================================
\section{System Model and Problem Formulation}
\label{sec:system}
% =============================================================================
\subsection{System Model}
\label{subsec:system_model}

We consider the downlink of a CF-mMIMO system operating in time-division duplex (TDD) mode. The network consists of $L$ APs, each equipped with $M$ antennas, that jointly serve $K$ single-antenna user equipments (UEs) over the same time-frequency resource, with $LM \gg K$~\cite{interdonato2020local}. All APs are connected to a central processing unit (CPU) via fronthaul links. 
% We assume independent Rayleigh fading channels and standard MMSE-based channel estimation with pilot reuse.
The channel between AP~$l$ and UE~$k$ is denoted by $\vh_{l,k} \in \mathbb{C}^{M \times 1}$ and follows $\vh_{l,k} \sim \mathcal{CN}(\mathbf{0}, \beta_{l,k}\mI_M)$, where $\beta_{l,k}$ represents the large-scale fading coefficient. During uplink training, each UE transmits a pilot sequence $\boldsymbol{\phi}_{i_k} \in \mathbb{C}^{\tau_p}$ with $\|\boldsymbol{\phi}_{i_k}\|^2 = \tau_p$, where $i_k$ is the pilot index and $\tau_p$ is the pilot length. We assume $\tau_p \le K$, and thus pilot reuse occurs. The set of UEs sharing the same pilot as UE~$k$ is denoted by $\mathcal{P}_k = \{t : i_t = i_k\}$, which leads to pilot contamination.
Let $p_k$ denote the uplink transmit power of UE~$k$. The MMSE estimate of $\vh_{l,k}$, denoted by $\hat{\vh}_{l,k}$, has zero mean and covariance $\gamma_{l,k}\mI_M$, where
% Let $p_k$ denote the uplink normalized transmit power of UE~$k$. The MMSE estimate of $\vh_{l,k}$ is denoted by $\hat{\mathbf{h}}_{l,k}$ and is obtained by correlating the received pilot signal with $\boldsymbol{\phi}_{i_k}$. The estimate has zero mean and covariance $\gamma_{l,k}\mI_M$, where
\begin{equation}
    \gamma_{l,k}
    = \frac{\tau_p p_k \beta_{l,k}^{2}}
           {\tau_p \sum_{t \in \mathcal{P}_k} p_t \beta_{l,t} + 1}.
    \label{eq:gamma}
\end{equation}
% The estimation error $\tilde{\mathbf{h}}_{l,k} = \vh_{l,k} - \hat{\mathbf{h}}_{l,k}$ is independent of $\hat{\mathbf{h}}_{l,k}$ and has covariance $(\beta_{l,k}-\gamma_{l,k})\mI_M$. 

% Each AP~$l$ employs local partial zero-forcing (PZF) precoding~\cite{interdonato2020local}. Specifically, the set of UEs is partitioned into a strong set $\mathcal{S}_l$ and a weak set $\mathcal{W}_l = \{1,\ldots,K\} \setminus \mathcal{S}_l$ based on the large-scale fading coefficients $\{\beta_{l,k}\}$. At AP $l$, full zero-forcing (ZF) is applied to UEs in $\mathcal{S}_l$, while maximum-ratio transmission (MRT) is used for UEs in $\mathcal{W}_l$. Since ZF is applied over pilot subspaces, if $k \in \mathcal{S}_l$, then all UEs sharing the same pilot as UE~$k$ also belong to $\mathcal{S}_l$. Let $\tau_{\mathcal{S}_l}$ denote the number of distinct pilots assigned to $\mathcal{S}_l$. To apply ZF, it is required that $M > \tau_{\mathcal{S}_l}$.

We assume that each AP~$l$ employs PZF precoding, as it provides very good performance~\cite{interdonato2020local}. The set of UEs is partitioned into a strong set $\mathcal{S}_l$ and a weak set $\mathcal{W}_l = \{1,\ldots,K\} \setminus \mathcal{S}_l$ based on the large-scale fading coefficients ${\beta_{l,k}}$. Zero-forcing (ZF) is applied to UEs in $\mathcal{S}_l$ to suppress interference within pilot subspaces, while MRT is used for UEs in $\mathcal{W}_l$ to reduce computational complexity. This structure allows interference suppression for UEs with strong channels while maintaining low complexity for the remaining UEs.
Since ZF is applied over pilot subspaces, if $k \in \mathcal{S}_l$, then all UEs sharing the same pilot as UE~$k$ belong to $\mathcal{S}_l$. Let $\tau_{\mathcal{S}_l}$ denote the number of distinct pilots assigned to $\mathcal{S}_l$, and the condition $M > \tau_{\mathcal{S}_l}$ is required for ZF.

The downlink transmit signal from AP~$l$ is $\vx_l = \sum_{k=1}^{K}\sqrt{\rho_{l,k}}\,\vw_{l,k} q_k,$ where $\rho_{l,k} \ge 0$ denotes the transmit power allocated by AP~$l$ to UE~$k$, $\vw_{l,k}$ is the precoding vector, and $q_k \sim \mathcal{CN}(0,1)$ is the data symbol. An achievable downlink SE for UE~$k$ under local PZF is given by
% \begin{equation}
$
    \mathrm{SE}_k
    = \frac{\tau_c - \tau_p}{\tau_c}
    \log_2\!\Bigl(1 + \mathrm{SINR}_k\Bigr),
    \label{eq:se}
$
% \end{equation}
where $\tau_c$ is the coherence block length. The corresponding signal-to-interference-plus-noise ratio (SINR) is given by \cite{interdonato2020local}
\begin{equation}
\mathrm{SINR}_k
=
\frac{
\Bigl(\sum\limits_{l=1}^{L}
  a_{lk}\sqrt{\rho_{l,k}}\Bigr)^2
}{
\sum\limits_{t\in\mathcal{P}_k\setminus\{k\}}
\Bigl(\sum\limits_{l=1}^{L}
  b_{l,k,t}\sqrt{\rho_{l,t}}\Bigr)^2
+\sum\limits_{t=1}^{K} \sum\limits_{l=1}^{L} d_{l,k,t} \rho_{l,t} 
+1
},
\label{eq:sinr}
\end{equation}
where $\delta_{l,k}$ equals one if $k \in \mathcal{S}_l$, and zero otherwise, $a_{l,k} \triangleq \sqrt{(M - \delta_{l,k}\tau_{\mathcal{S}_l})\gamma_{l,k}}$, $b_{l,k,t} \triangleq \sqrt{(M - \delta_{l,t}\tau_{\mathcal{S}_l})\gamma_{l,k}}$ and $d_{l,k,t} \triangleq \beta_{l,k} - \delta_{l,k}\delta_{l,t}\gamma_{l,k}$.
% \begin{figure*}[!b]
% \normalsize
% \begin{equation}
% \mathrm{SINR}_k
% =
% \frac{
% \Bigl(\sum_{l=1}^{L}
%   \sqrt{(M-\delta_{l,k}\tau_{\mathcal{S}_l})\rho_{l,k}\,\gamma_{l,k}}\Bigr)^2
% }{
% \sum_{t\in\mathcal{P}_k\setminus\{k\}}
% \Bigl(\sum_{l=1}^{L}
%   \sqrt{(M-\delta_{l,t}\tau_{\mathcal{S}_l})\rho_{l,t}\,\gamma_{l,k}}\Bigr)^2
% +\sum_{t=1}^{K}\sum_{l=1}^{L}\rho_{l,t}
%   \bigl(\beta_{l,k}-\delta_{l,k}\delta_{l,t}\gamma_{l,k}\bigr)
% +1
% }.
% \label{eq:sinr}
% \end{equation}
% \end{figure*}

\subsection{Energy Efficiency Problem Formulation}
\label{subsec:problem}

The total power consumption of the considered CF-mMIMO system is modeled as \cite{ngo2017total,mai2022apg}
\begin{equation}
    P_{\mathrm{total}}
    = \sum_{l=1}^{L}\frac{N_0}{\alpha_l}\sum_{k=1}^{K}\rho_{l,k}
    + \bar{P}_{\mathrm{fix}}
    + B \bar{P}_{\mathrm{bt}} \sum_{k=1}^{K} \mathrm{SE}_k,
    \label{eq:Ptotal}
\end{equation}
where $\alpha_l \in (0,1]$ denotes the power-amplifier efficiency at AP~$l$, $N_0$ is the noise power, and $B$ is the system bandwidth. The first term represents the transmit power, where $N_0$ converts the normalized power $\rho_{l,k}$ to physical units. The scalar $\bar{P}_{\mathrm{fix}} \triangleq \sum_{l=1}^{L} (M P_{\mathrm{tc},l} + P_{0,l})$ captures the static circuit and fronthaul power across all APs, where $P_{\mathrm{tc},l}$ is the per-antenna circuit power and $P_{0,l}$ is the fixed fronthaul power at AP~$l$. The last term models the traffic-dependent fronthaul power, where $P_{\mathrm{bt},l}$ denotes the traffic-dependent coefficient at AP~$l$, and $\bar{P}_{\mathrm{bt}} \triangleq \sum_{l=1}^{L} P_{\mathrm{bt},l}$.

The EE of the network, measured in bits per joule, is defined as the ratio between the total achievable data rate and the total power consumption, given by
% \begin{equation}
$
    \mathrm{E_e}
    = \frac{B\sum_{k=1}^{K}\mathrm{SE}_k}{P_{\mathrm{total}}}.
    \label{eq:EE}
$
% \end{equation}
The EE maximization problem is formulated as 
\begin{subequations}
\label{eq:P1}
\begin{align}
    (\mathcal{P}_1):\quad
    \max_{\{\rho_{l,k}\}} \quad
    & \mathrm{E_e} \label{eq:P1_obj} \\
    \text{s.t.} \quad
    & \mathrm{SE}_k \ge S^{\mathrm{min}}_k, \quad k = 1,\ldots,K, \label{eq:P1_qos} \\
    & \sum_{k=1}^{K}\rho_{l,k} \le \rho_l^{\max}, \quad l = 1,\ldots,L,
      \label{eq:P1_power} \\
    & \rho_{l,k} \ge 0, \quad \forall l, k. \label{eq:P1_nonneg}
\end{align}
\end{subequations}
The QoS constraints in~\eqref{eq:P1_qos} enforce minimum SE requirements $S^{\mathrm{min}}_k$ for all UEs, while~\eqref{eq:P1_power} limits the transmit power at each AP. 
Problem $(\mathcal{P}_1)$ is non-convex due to the fractional structure of the objective function and the coupling of the optimization variables in the SINR expressions% , which appear in both the objective and the QoS constraints
. As a result, the problem is challenging to solve using conventional optimization techniques. In the next section, we reformulate $(\mathcal{P}_1)$ to enable efficient gradient-based optimization and develop an APG-based solution together with a deep-unfolded model.
% =============================================================================
\section{Proposed Deep-Unfolded APG Method }
\label{sec:proposed}
% =============================================================================

To address $(\mathcal{P}_1)$, we first develop an APG-based optimization framework and then construct a deep-unfolded network by mapping APG iterations into trainable deep neural network (DNN) layers.%, replacing the iterative procedure with a model with a fixed number of layers in which the step sizes and penalty parameters are learned.
% -----------------------------------------------------------------------------
\subsection{APG Optimization Framework}
\label{subsec:apg}
% -----------------------------------------------------------------------------
\subsubsection{Problem Reformulation}
The optimization problem $(\mathcal{P}_1)$ is challenging to solve directly due to its non-convex objective and QoS constraints. To enable the application of the APG method, we first reformulate the problem. Specifically, we introduce variable $\theta_{l,k} = \sqrt{\rho_{l,k}} \ge 0$ to avoid square-root terms in the SINR and QoS expressions and facilitate more efficient projection and gradient updates. Under this transformation, the per-AP power constraint in~\eqref{eq:P1_power} becomes
$
\sum_{k=1}^{K} \theta_{l,k}^2 \le \rho_l^{\max}.
$
Let $\boldsymbol{\theta}_l \triangleq [\theta_{l,1}, \ldots, \theta_{l,K}]^T$ denote the power allocation vector for AP $l$, and let $\boldsymbol{\theta} \in \mathbb{R}^{L \times K}$ stack all such vectors across APs. The feasible sets for the power allocation variables $\boldsymbol{\theta}_l$ and $\boldsymbol{\theta}$ are defined as
$
\mathcal{C}_l
= \left\{\boldsymbol{\theta}_l \in \mathbb{R}^K :
\theta_{l,k} \ge 0,\ \forall k,\ 
\sum_{k=1}^{K} \theta_{l,k}^2 \le \rho_l^{\max} \right\}
$
and
$
\mathcal{C}
= \left\{ \boldsymbol{\theta} \in \mathbb{R}^{L \times K} :
\boldsymbol{\theta}_l \in \mathcal{C}_l,\ \forall l \right\},
$
respectively.

Let 
$
A_{k,t} \triangleq \sum_{l=1}^{L} b_{l,k,t}\,\theta_{l,t}, \quad
A_k \triangleq \sum_{l=1}^{L} a_{l,k}\,\theta_{l,k},
$
$
I_k \triangleq \sum_{t \in \mathcal{P}_k \setminus\{k\}} A_{k,t}^2
+ \sum_{t=1}^{K}\sum_{l=1}^{L} d_{l,k,t}\,\theta_{l,t}^2 + 1.
$
With these definitions, the SINR can then be written as $\mathrm{SINR}_k = A_k^2 / I_k$, and the QoS constraint in~\eqref{eq:P1_qos} is equivalent to
\begin{equation}
    A_k \ge \sqrt{\bar{S}^{\mathrm{min}}_{k}}\,\sqrt{I_k},
\end{equation}
where $\bar{S}^{\mathrm{min}}_{k} \triangleq 2^{S^{\mathrm{min}}_k\tau_c/(\tau_c-\tau_p)} - 1$.
The QoS constraints introduce non-convex coupling between the variables and cannot be handled directly within the APG updates. Therefore, we incorporate them into the objective using a penalty formulation. Specifically, we define
\begin{equation}
    g_k(\boldsymbol{\theta})
    \triangleq \sqrt{\bar{S}^{\mathrm{min}}_{k}\,I_k} - A_k,
    \label{eq:gk}
\end{equation}
which is non-positive when the QoS constraint is satisfied. The reformulated problem is given by\begin{equation}
    (\mathcal{P}_2):
    \max_{\boldsymbol{\theta} \in \mathcal{C}}
    \quad
    f_\xi(\boldsymbol{\theta})
    \triangleq
    \mathrm{EE}(\boldsymbol{\theta})
    - \xi \sum_{k=1}^{K} \bigl[\max(0,\,g_k(\boldsymbol{\theta}))\bigr]^2,
\end{equation}
where $\xi > 0$ is a penalty parameter that controls the strength of constraint enforcement. 
% The penalty term enforces the QoS constraints, with larger values of $\xi$ imposing stronger penalties on constraint violations.

\subsubsection{APG Update Rules and Projection}

To solve $(\mathcal{P}_2)$, we apply the APG method~\cite{mai2022apg}. %Unlike the projected gradient method, which determines each step using only the current iterate, APG also uses the change between consecutive iterates to guide the update. This allows the method to continue moving along the direction of the previous update instead of relying solely on local gradient information, resulting in faster convergence.
Unlike the projected gradient method, which computes each update based solely on the current iterate, APG additionally exploits the difference between successive iterates to capture the recent search direction. This allows the method to continue moving along the direction of the previous update instead of relying solely on local gradient information, thereby accelerating convergence.

Let $\boldsymbol{\theta}^{(n)} \in \mathcal{C}$ denote the power allocation vector at iteration $n$. At each iteration $n \geq 1$, the point $\vy^{(n)}$, which incorporates information from previous updates, is constructed as
\begin{equation}
    \vy^{(n)}
    = \boldsymbol{\theta}^{(n)}
    + \frac{s_{n-1}}{s_n}
    \bigl(\vz^{(n)} - \boldsymbol{\theta}^{(n)}\bigr)
    + \frac{s_{n-1} - 1}{s_n}
    \bigl(\boldsymbol{\theta}^{(n)} - \boldsymbol{\theta}^{(n-1)}\bigr),
    \label{eq:accelarated_point}
\end{equation}
where $s_n = (1 + \sqrt{1 + 4s_{n-1}^2})/2$ controls the contribution of previous updates. The term $\boldsymbol{\theta}^{(n)} - \boldsymbol{\theta}^{(n-1)}$ represents the difference between consecutive iterates and allows the update to follow the direction of the previous step. The initialization is given by $\boldsymbol{\theta}^{(0)} \in \mathcal{C}$, $\boldsymbol{\theta}^{(1)} = \vz^{(1)} = \boldsymbol{\theta}^{(0)}$, and $s_0 = s_1 = 1$.
% To balance progress and reliable improvement, two updates are computed \Nhan{I commented many times already, don't write like this. For this sentence, it can be simply written as ``To balance progress and reliable improvement, $\vz$ and $\vv$ are updated as:"} at each iteration as
To balance progress and reliable improvement, $\vz^{(n+1)}$ and $\vv^{(n+1)}$ are updated as
\begin{align}
    \vz^{(n+1)}
    &= \mathcal{P}_{\mathcal{C}}\!\bigl(
       \vy^{(n)} + \alpha_{y,n}\,\nabla f_\xi(\vy^{(n)})
       \bigr), \label{eq:z}
       \\
    \vv^{(n+1)}
    &= \mathcal{P}_{\mathcal{C}}\!\bigl(
       \boldsymbol{\theta}^{(n)} + \alpha_{\theta,n}\,
       \nabla f_\xi(\boldsymbol{\theta}^{(n)})
       \bigr), \label{eq:v}
\end{align}
where $\alpha_{y,n}$ and $\alpha_{\theta,n}$ are step sizes. The update $\vz^{(n+1)}$ incorporates information from previous iterations and promotes faster progress, while $\vv^{(n+1)}$ relies only on the current gradient and ensures stable improvement.
The update yielding the higher objective value is selected as the next iterate, i.e.,
\begin{equation}
    \label{eq:selection_power}
    \boldsymbol{\theta}^{(n+1)}
    =
    \begin{cases}
        \vz^{(n+1)}, & f_\xi(\vz^{(n+1)}) \ge f_\xi(\vv^{(n+1)}), \\
        \vv^{(n+1)}, & \text{otherwise}.
    \end{cases}
\end{equation}
% which ensures non-decreasing objective values over iterations. 
The projection $\mathcal{P}_{\mathcal{C}}(\cdot)$ enforces the per-AP power constraints. Since $\mathcal{C}$ is separable across APs, the projection is applied independently to each $\boldsymbol{\theta}_l$. For $\vu_l \in \mathbb{R}^K$, the projection onto $\mathcal{C}_l$ is
\begin{equation}
    \mathcal{P}_{\mathcal{C}_l}(\vu_l)
    =
    \min\!\left(1,\,
        \frac{\sqrt{\rho_l^{\max}}}{\|\tilde{\vu}_l\|}
    \right)
    \tilde{\vu}_l,
    \label{eq:proj}
\end{equation}
where $\tilde{\vu}_l$ is obtained by setting negative entries of $\vu_l$ to zero.
% , i.e., $\tilde{u}_{l,k} = \max(u_{l,k}, 0)$.

\subsubsection{Gradient of the Objective Function}
\label{subsubsec:gradient}

The APG updates in~\eqref{eq:z}–\eqref{eq:v} require the gradient of $f_\xi(\boldsymbol{\theta})$. We present its closed-form gradient in the following theorem. %To enable efficient evaluation of these updates, we derive closed-form expressions for the gradient.
% , which avoid numerical differentiation and allow low-complexity implementation.
% Define the penalty term as
% $
%     \Psi(\boldsymbol{\theta}) \triangleq \sum_{k=1}^{K} \bigl[\max(0, g_k(\boldsymbol{\theta}))\bigr]^2.
% $
% The gradient of $f_\xi(\boldsymbol{\theta})$ is given by the following result.
\begin{theorem}
\label{thm:gradient}
Let $u \triangleq \sum_{k=1}^{K} \mathrm{SE}_k$ denote the sum SE and let $\Psi(\boldsymbol{\theta}) \triangleq \sum_{k=1}^{K} \bigl[\max(0, g_k(\boldsymbol{\theta}))\bigr]^2$. Furthermore, we rewrite the total power consumption as $P_{\mathrm{total}} = \tilde{P} + B\bar{P}_{\mathrm{bt}}\,u,$ where
% \begin{equation}
$
    \tilde{P} \triangleq \bar{P}_{\mathrm{fix}} + \sum_{l=1}^{L}\sum_{k=1}^{K}({N_0}/{\alpha_l})\theta_{l,k}^2.
$
% \end{equation}
Then, the gradient of $f_\xi(\boldsymbol{\theta})$ with respect to $\boldsymbol{\theta} \in \mathbb{R}^{L \times K}$ is given by
\begin{equation}
    \nabla f_\xi = \nabla \mathrm{E_e} - \xi\,\nabla \Psi,
\end{equation}
where
\begin{align}
    \nabla \mathrm{E_e}
    &=
    \frac{B}{P_{\mathrm{total}}^2}
    \Bigl(\tilde{P}\,\nabla u
    - 2u N_0\,\boldsymbol{\alpha}^{-1} \odot \boldsymbol{\theta}\Bigr),
    \label{eq:grad_Ee}\\
    \nabla u
    &=
    \sum_{k=1}^{K}
    \frac{\tau_c-\tau_p}{\tau_c \ln 2 \,(A_k^2 + I_k)}
    \bigl(\nabla A_k^2 - \mathrm{SINR}_k\,\nabla I_k\bigr),
    \label{eq:grad_u}\\
    \nabla \Psi
    &=
    \sum_{k:\, g_k > 0}
    2\,g_k\,\nabla g_k.
    \label{eq:grad_Psi}
\end{align}
Here $[\boldsymbol{\alpha}^{-1}]_{l,k} = 1/\alpha_l$, and $\odot$ denotes element-wise multiplication. 
Furthermore, $A_k$, $I_k$, and $g_k$ are scalar functions of $\boldsymbol{\theta}$, while their gradients $\nabla A_k^2$, $\nabla I_k$, and $\nabla g_k$ are matrices in $\mathbb{R}^{L \times K}$ whose entries are defined in~\eqref{eq:app_grad_DS}, \eqref{eq:app_grad_Ik}, and \eqref{eq:app_grad_gk}, respectively.

% \begin{align}
%     [\nabla A_k^2]_{m,t}
%     &= 2\,a_{m,k}\,A_k \cdot \mathbf{1}_{\{t=k\}},
%     \label{eq:grad_DS}\\
%     [\nabla I_k]_{m,t}
%     &= 2\,b_{m,k,t}\,A_{k,t}\cdot\mathbf{1}_{\{t\in\mathcal{P}_k\setminus\{k\}\}}
%     + 2\,d_{m,k,t}\,\theta_{m,t},
%     \label{eq:grad_Ik}\\
%     [\nabla g_k]_{m,t}
%     &= \frac{\sqrt{\bar{S}^{\mathrm{min}}_{k}}}{\sqrt{I_k}}
%     \Bigl(b_{m,k,t}\,A_{k,t}\cdot\mathbf{1}_{\{t\in\mathcal{P}_k\setminus\{k\}\}}
%     + d_{m,k,t}\,\theta_{m,t}\Bigr)
%     \nonumber\\
%     &\quad - a_{m,k}\cdot\mathbf{1}_{\{t=k\}},
%     \label{eq:grad_gk}
% \end{align}
% where $\mathbf{1}_{\{\cdot\}}$ takes the value $1$ if the condition is satisfied and $0$ otherwise.
\end{theorem}
\begin{IEEEproof}
See Appendix~\ref{app:gradient_derivation}.
\end{IEEEproof}

With the closed-form expressions in Theorem~\ref{thm:gradient} and the projection in~\eqref{eq:proj}, the gradient and projection steps in APG can be evaluated explicitly. However, the performance and convergence  of the APG method depend on the step sizes $\alpha_{y,n}$ and $\alpha_{\theta,n}$. Manual tuning may lead to suboptimal convergence, while line search with backtracking requires repeated objective evaluations and increases computational complexity, particularly in massive MIMO systems. This results in high and variable computational cost. To address this issue, we propose mapping the APG iterations into a fixed number of neural network layers, which allows  learning $\alpha_{y,n}$ and $\alpha_{\theta,n}$ from data training, as elaborated next.

\subsection{Proposed APG-based Unfolding Network}
\label{subsec:unfolded}

We employ a deep-unfolded network with $T$ layers, each performs the update of one APG iteration, as  in~\eqref{eq:accelarated_point}–\eqref{eq:selection_power}. The step sizes $\alpha_{y,n}$ and $\alpha_{\theta,n}$ and penalty parameter $\xi$ are treated as learnable variables, eliminating line search and yielding fixed inference complexity.

\subsubsection{Network Structure}

Fig.~\ref{fig:apg_architecture} illustrates the structure of one layer of the proposed deep-unfolded APG network. Each layer corresponds to an iteration of the APG method in \eqref{eq:accelarated_point}--\eqref{eq:selection_power}. 
Within a layer, $\vy^{(t)}$ is formed from $\boldsymbol{\theta}^{(t)}$, $\boldsymbol{\theta}^{(t-1)}$, and $\vz^{(t)}$. Then, the updates $\vz^{(t+1)}$ and $\vv^{(t+1)}$ are computed based on $\vy^{(t)}$ and $\boldsymbol{\theta}^{(t)}$, respectively. Each update evaluates the gradient of $f_{\xi_t}$, scales it with the learnable step sizes $\alpha_{y,t}$ and $\alpha_{\theta,t}$, and projects the result onto the feasible set $\mathcal{C}$.
The key difference from the APG method lies in how the next iterate is obtained. In \eqref{eq:selection_power}, the update with the larger objective value is selected through a comparison step. In the proposed network, this step is replaced by a learnable weighted combination
\begin{equation}
    \boldsymbol{\theta}^{(t+1)}
    = w_t\,\vz^{(t+1)}
    + (1 - w_t)\,\vv^{(t+1)},
    \label{eq:weighted_combination}
\end{equation}
where $w_t \in (0,1)$ is a learnable parameter. This replacement eliminates the non-differentiable selection step in~\eqref{eq:selection_power}, enabling end-to-end training, and ensures $\boldsymbol{\theta}^{(t+1)} \in \mathcal{C}$ since $\vz^{(t+1)}, \vv^{(t+1)} \in \mathcal{C}$.

\begin{figure}[htbp]
    \centering
    \includegraphics[width=\linewidth]{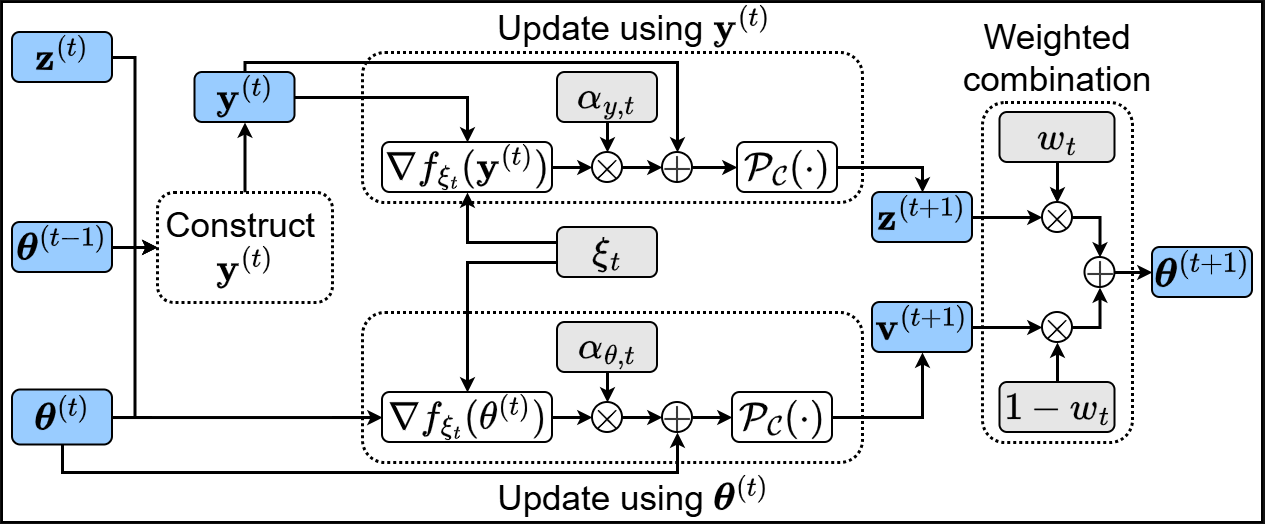}
    \caption{Illustration of the $t$-th layer of the proposed unfolded PGA model.% \Nhan{the text in the figure is too small to read.}
    }
    \label{fig:apg_architecture}
\end{figure}
% =========================================================

% =========================================================

\subsubsection{Learnable Parameters and Training}
% In the APG method, the step sizes $\alpha_{y,t}$ and $\alpha_{\theta,t}$ are determined via line search, the update is selected by comparison in~\eqref{eq:selection_power}, and the penalty parameter $\xi$ is adjusted across iterations according to a predefined schedule. \Nhan{we don't need to repeat this all the time.} 
In the proposed network, the step sizes $\alpha_{y,n}$, $\alpha_{\theta,n}$, penalty parameter $\xi$ and  the update in~\eqref{eq:selection_power} are replaced by a set of learnable parameters
\begin{equation}
    \boldsymbol{\Phi}
    =
    \bigl\{
    \tilde{\alpha}_{y,t},\, \tilde{\alpha}_{\theta,t},\, \xi_t,\, w_t
    \bigr\}_{t=1}^{T},
\end{equation}
where $\tilde{\alpha}_{y,t}, \tilde{\alpha}_{\theta,t}, \xi_t > 0$ and $w_t \in (0,1)$. 
% This replacement eliminates line search and yields fixed per-layer computational cost.
The step sizes $\alpha_{y,t}$ and $\alpha_{\theta,t}$ are constructed by applying learnable scaling factors $\tilde{\alpha}_{y,t}$ and $\tilde{\alpha}_{\theta,t}$ to the Barzilai–Borwein (BB) estimates, given by \cite{mai2022apg} % \Nhan{any reference?}
\begin{align}
    \alpha_{\mathrm{BB},z}^{(t)}
    =
    \frac{\|\vy^{(t)} - \vy^{(t-1)}\|^2}
         {\bigl|\bigl\langle
             \vy^{(t)} - \vy^{(t-1)},\;
             \nabla f_{\xi_t}(\vy^{(t)}) - \nabla f_{\xi_t}(\vy^{(t-1)})
         \bigr\rangle\bigr|},
    \\
    \alpha_{\mathrm{BB},v}^{(t)}
    =
    \frac{\|\boldsymbol{\theta}^{(t)} - \boldsymbol{\theta}^{(t-1)}\|^2}
         {\bigl|\bigl\langle
             \boldsymbol{\theta}^{(t)} - \boldsymbol{\theta}^{(t-1)},\;
             \nabla f_{\xi_t}(\boldsymbol{\theta}^{(t)}) - \nabla f_{\xi_t}(\boldsymbol{\theta}^{(t-1)})
         \bigr\rangle\bigr|}.
\end{align}
Since the BB estimate requires two consecutive iterates, it is not defined at $t=1$. The resulting step sizes are given by
\begin{equation}
\alpha_{y,t} =
\begin{cases}
\tilde{\alpha}_{y,t}, & t=1,\\
\tilde{\alpha}_{y,t}\,\alpha_{\mathrm{BB},z}^{(t)}, & t>1,
\end{cases}
\quad
\alpha_{\theta,t} =
\begin{cases}
\tilde{\alpha}_{\theta,t}, & t=1,\\
\tilde{\alpha}_{\theta,t}\,\alpha_{\mathrm{BB},v}^{(t)}, & t>1.
\label{eq:update_step}
\end{cases}
\end{equation}

The network is trained to maximize EE while enforcing the QoS constraints. For each layer $t$, the loss is defined as
\begin{equation}
    \mathcal{L}^{(t)}
    =
    -\mathrm{E_e}\!\bigl(\boldsymbol{\theta}^{(t)}\bigr)
    + \xi_{\mathrm{fix}}\,\Psi\!\bigl(\boldsymbol{\theta}^{(t)}\bigr),
\end{equation}
where $\xi_{\mathrm{fix}} > 0$ penalizes QoS violations. The parameter $\xi_t$ is used within each layer for gradient computation, while $\xi_{\mathrm{fix}}$ is fixed and used only in the training loss.
A progressive learning strategy~\cite{shlezinger2025deep} is adopted, where layers are optimized sequentially from $t=1$ to $T$. At layer $t$, only its
% \Nhan{its here mean ``layer $T$". Is this what you mean?} 
parameters are updated by minimizing $\mathcal{L}^{(t)}$, while earlier layers are kept fixed, which improves training stability.

% =========================================================
\subsubsection{Inference Procedure}

The inference procedure of the proposed unfolded APG-based power control scheme is summarized in Algorithm~\ref{alg:inference}. For a given channel realization, the power allocation is obtained via a $T$-layer forward pass using the trained parameters $\boldsymbol{\Phi}^{\star}$, as illustrated in Fig.~\ref{fig:apg_architecture}. The network is initialized using the heuristic channel-dependent (HCD) power allocation~\cite{interdonato2020local}, defined as
\begin{equation}
    \rho_{l,k}^{(0)}
    =
    \frac{\gamma_{l,k}}{\sum_{i=1}^{K}\gamma_{l,i}} \rho_l^{\max},
    \qquad
    \theta_{l,k}^{(0)} = \sqrt{\rho_{l,k}^{(0)}},
    \quad \forall l,k,
    \label{eq:s32_hcd_init}
\end{equation}
which allocates power proportionally to the large-scale fading coefficients $\{\gamma_{l,k}\}$.
% and provides a feasible starting point with low computational cost
% The final output $\boldsymbol{\theta}^{(T)}$ satisfies all constraints, since $\mathcal{P}_{\mathcal{C}}$ enforces feasibility at each update and the weighted combination in \eqref{eq:weighted_combination} preserves feasibility.

\setlength{\textfloatsep}{7pt}
\begin{algorithm}[t]
\small
\caption{Unfolded APG-based power control}
\label{alg:inference}
\LinesNumbered
\KwIn{$\{\beta_{l,k},\gamma_{l,k},\delta_{l,k}\}$, $\{S^{\mathrm{min}}_k\}$, $\{\rho_l^{\max}\}$, and trained parameters $\boldsymbol{\Phi}^{\star} = \{\tilde{\alpha}_{y,t}, \tilde{\alpha}_{\theta,t}, \xi_t, w_t\}_{t=1}^T$.}
\KwOut{Power allocation $\boldsymbol{\theta}^{(T+1)} \in \mathcal{C}$.}
\textbf{Initialization:} \\
Precompute $\{s_t\}_{t=1}^T$\;
Initialize $\boldsymbol{\theta}^{(0)}$ using HCD in~\eqref{eq:s32_hcd_init}\;
\For{$t = 1$ \KwTo $T$}{
    Compute $\vy^{(t)}$ using~\eqref{eq:accelarated_point}\;
    Compute $\nabla f_{\xi_t}(\vy^{(t)})$ and $\nabla f_{\xi_t}(\boldsymbol{\theta}^{(t)})$ using Theorem~\ref{thm:gradient}\;
    Compute step sizes $\alpha_{y,t}$ and $\alpha_{\theta,t}$ using~\eqref{eq:update_step}\;
    Update $\vz^{(t+1)}$ and $\vv^{(t+1)}$ using~\eqref{eq:z}-\eqref{eq:v}\;
    Update $\boldsymbol{\theta}^{(t+1)}$ 
    using~\eqref{eq:weighted_combination}\;
}
\Return $\boldsymbol{\theta}^{(T+1)}$\;
\end{algorithm}
% =========================================================
\subsubsection{Computational Complexity}

We herein analyze the computational complexity of Algorithm~\ref{alg:inference}. In each layer, the dominant cost arises from evaluating the gradients $\nabla f_{\xi_t}$, which has complexity $\mathcal{O}(L K^2)$ due to summations over AP and UE indices. The projection and combination steps have complexity $\mathcal{O}(L K)$. Hence, the per-layer complexity is $\mathcal{O}(L K^2)$, and the overall complexity of the $T$-layer network is $\mathcal{O}(T L K^2)$. The HCD initialization has complexity $\mathcal{O}(L K)$ and does not affect the overall order.
In contrast, the APG algorithm has complexity $\mathcal{O}(I_P I_{\mathrm{APG}} L K^2)$, where $I_P$ and $I_{\mathrm{APG}}$ denote the numbers of outer and inner iterations, respectively~\cite{mai2022apg}. The SCA method in~\cite{ngo2017total} has complexity $\mathcal{O}(I_{\mathrm{SCA}} \sqrt{L + K} L^3 K^4)$, where $I_{\mathrm{SCA}}$ is the number of iterations. In both cases, the iteration counts depend on the channel realizations and convergence criteria, resulting in variable computational cost, whereas the proposed method has fixed complexity determined by $T$.

\section{Numerical Results}

\begin{figure*}[t!]
    \centering
    \begin{subfigure}[b]{0.32\textwidth}
        \centering
        \includegraphics[width=\textwidth]{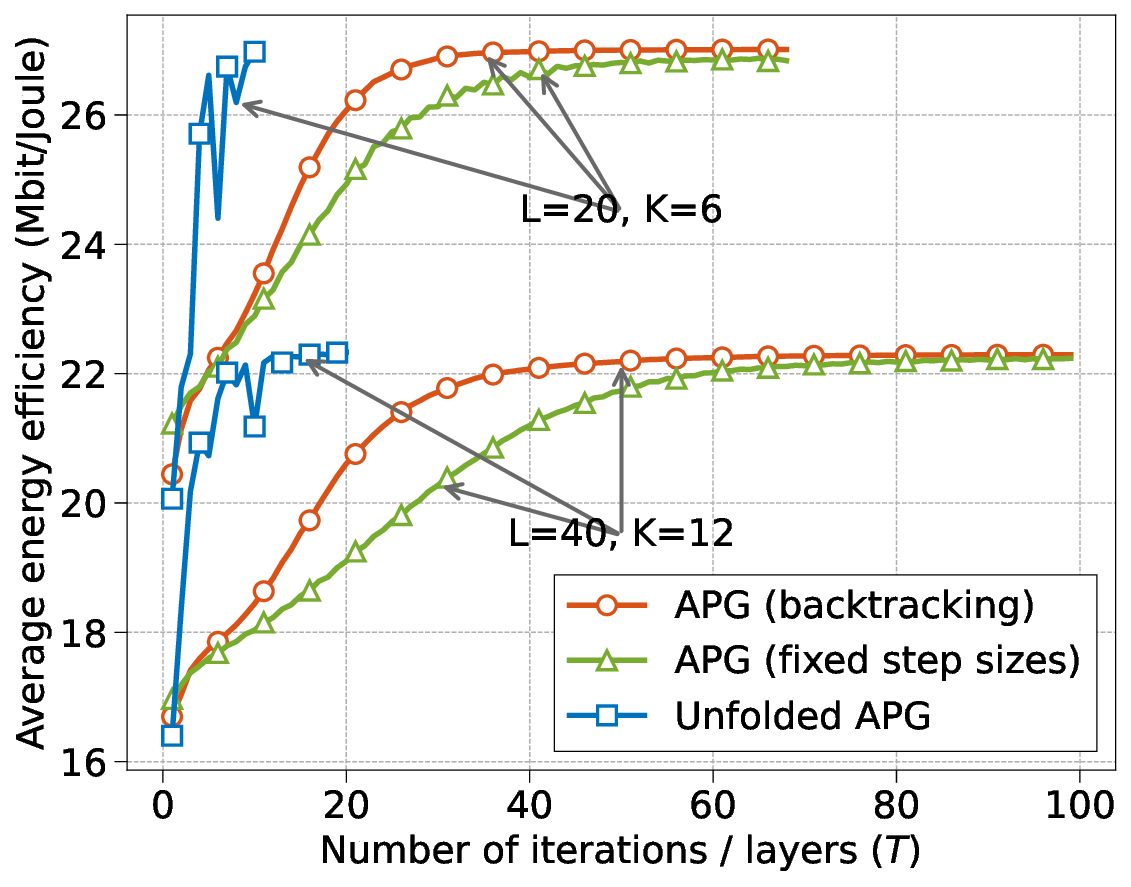}
        \caption{Convergence profile}
        \label{fig:convergence}
    \end{subfigure}
    \hfill
    \begin{subfigure}[b]{0.32\textwidth}
        \centering
        \includegraphics[width=\textwidth]{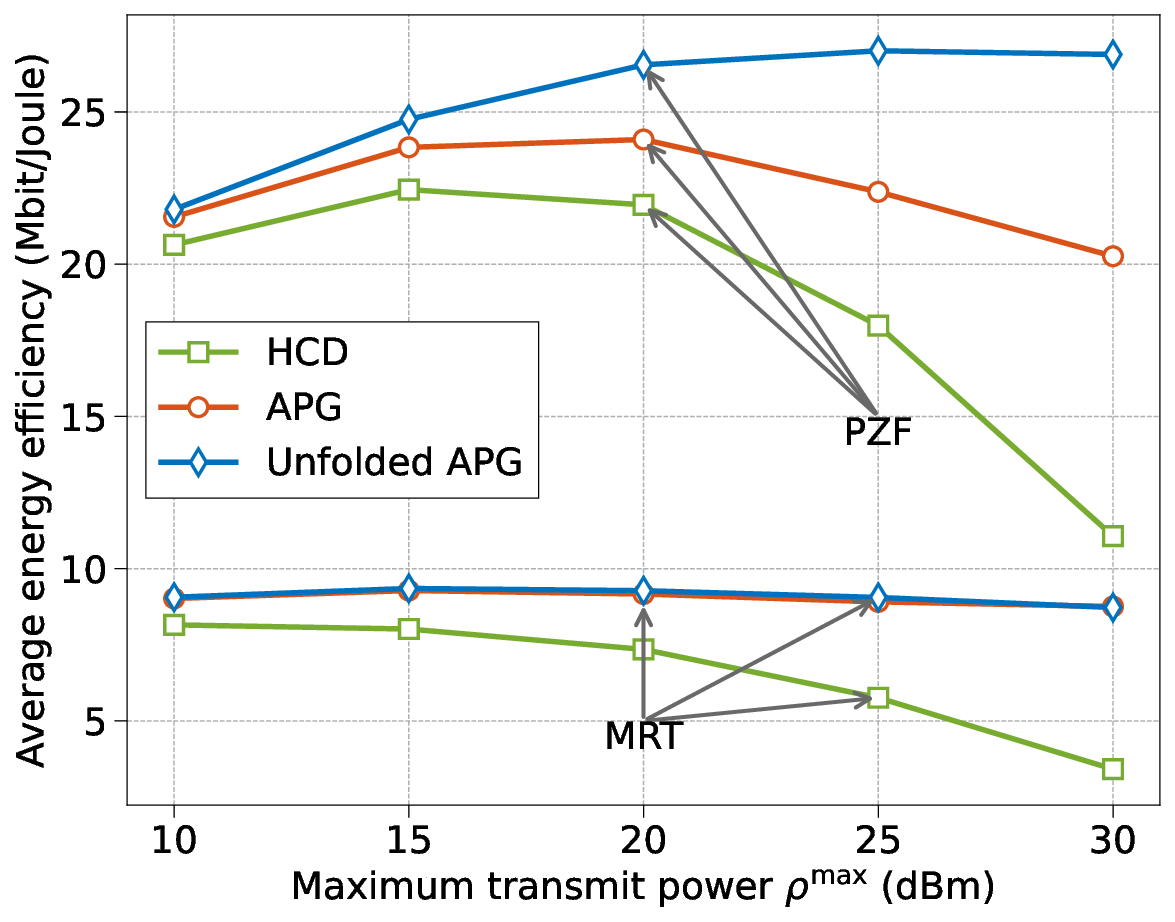}
        \caption{EE versus maximum transmit power.}
        \label{fig:ee_vs_rho}
    \end{subfigure}
    \hfill
    \begin{subfigure}[b]{0.32\textwidth}
        \centering
        \includegraphics[width=\textwidth]{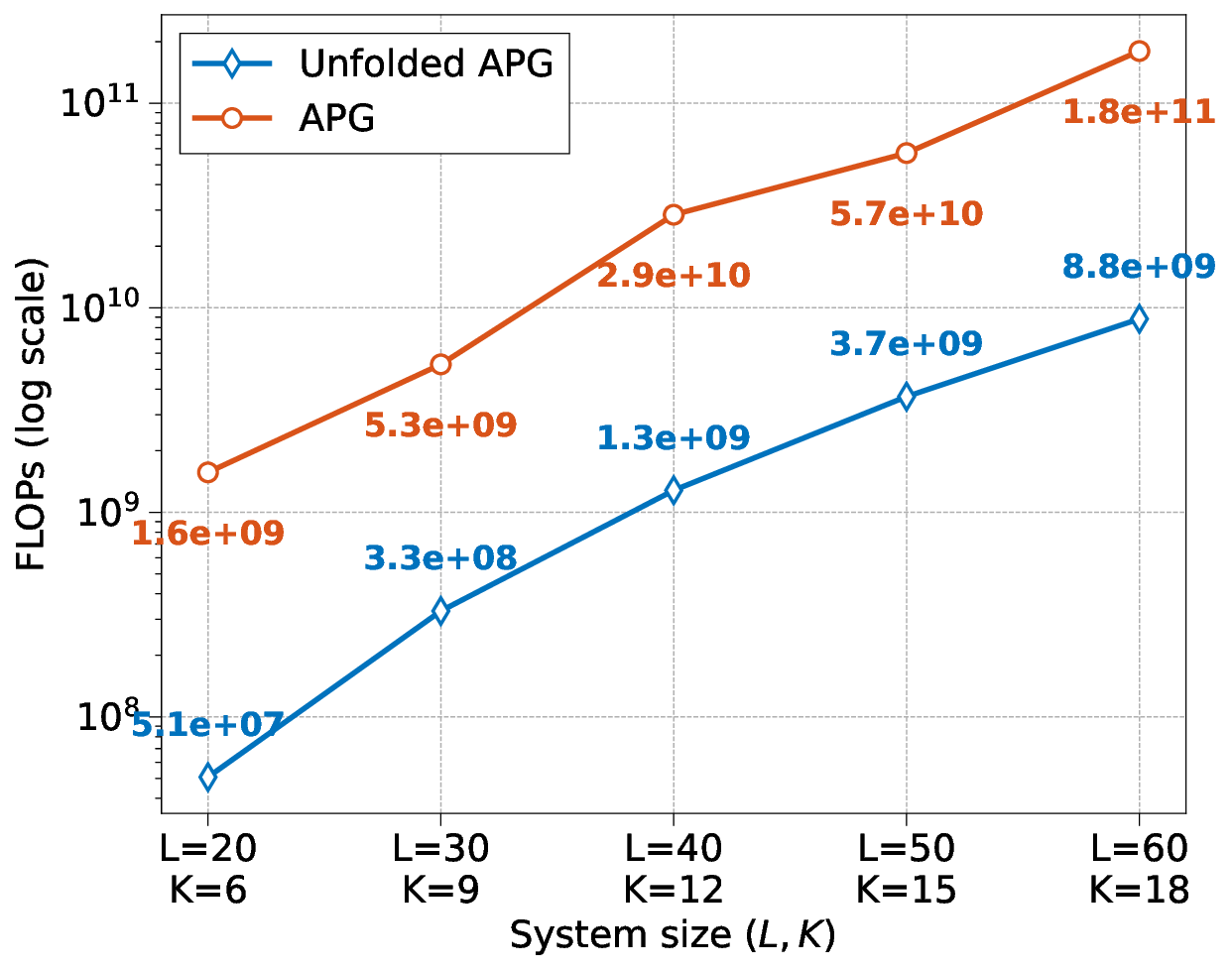}
        \caption{Computational cost versus system size $(L, K)$.}
        \label{fig:computation}
    \end{subfigure}
    \caption{Convergence, EE, and computational cost of the proposed unfolded APG scheme compared with benchmarks.}
    \label{fig:all_results}
\end{figure*}
This section evaluates the performance of the proposed deep-unfolded power control method. Unless otherwise specified, we consider a downlink CF-mMIMO system with $L=20$ APs, each equipped with $M=4$ antennas, serving $K=6$ single-antenna UEs. The APs and UEs are randomly distributed in a square area of side length $D=1000$ m using the wrap-around scheme.

The large-scale fading coefficients are modeled as
$
\beta_{l,k} = \mathsf{PL}_{l,k} \cdot 10^{\frac{\sigma_\text{sh} z_{l,k}}{10}},
$
where $\sigma_\text{sh}=4$ dB, and $z_{l,k}\sim\mathcal{N}(0,1)$ follows the correlation model in~\cite{interdonato2020local}. The path loss is given by
$
\mathsf{PL}_{l,k}\,[\mathrm{dB}] = -30.5 - 36.7\log_{10}\!\left(\frac{d_{l,k}}{1\,\mathrm{m}}\right),
$
where $d_{l,k}$ denotes the distance between AP~$l$ and UE~$k$, with antenna heights of $10$ m and $1.5$ m for the AP and UE, respectively~\cite{interdonato2020local}.
The system parameters are set as follows: bandwidth $B=20$ MHz, coherence block length $\tau_c=200$,  pilot length $\tau_p=5$, QoS target $S^{\mathrm{min}}_k=1$ bit/s/Hz, and noise power $N_0$ corresponding to a noise figure of $9$ dB~\cite{ngo2017total, mai2022apg}. Pilots are randomly assigned to the UEs. For the EE model in~\eqref{eq:Ptotal}, we set $\alpha_l=0.4$, $P_{\mathrm{tc},l}=0.2$ W, $P_{\mathrm{bt},l}=0.25$~W/(Gbit/s), and $P_{0,l}=0.825$~W~\cite{ngo2017total, mai2022apg}. 
The unfolded network has depth $T=10$ and is trained using Adam optimizer. The dataset consists of 800 training, 100 validation, and 100 testing channel realizations. Training is performed for 1000 epochs with mini-batch size 32 and learning rate 0.01. The parameters $\lambda_0$ and $\xi_{\mathrm{fix}}$ are set to 0.3 and 10, respectively.

To evaluate the proposed \textit{unfolded APG} method, we consider the heuristic \textit{HCD} method~\cite{interdonato2020local,bjornson2020scalable}, a backtracking-based \textit{APG} method in which the step sizes are determined via backtracking line search~\cite{mai2022apg}, and an \textit{APG} variant with fixed step sizes set to the average values obtained from the backtracking-based APG. All methods are evaluated under both PZF and MRT precoding schemes.

Fig.~\ref{fig:convergence} shows the convergence behavior of the proposed unfolded APG and the iterative APG methods for $(L=20, K=6)$ and $(L=40, K=12)$. The results indicate that the objective value of the proposed unfolded APG increases rapidly and attains a satisfactory EE performance, which is slightly superior to that of the APG-based counterparts at convergence. The backtracking-based APG converges within approximately 30 iterations, whereas the fixed-step APG requires more than 60 iterations to reach a stable solution. Compared with the fixed-step variant, APG with backtracking converges faster and achieves higher EE, at the expense of additional computational overhead due to the line search procedure. In contrast, the proposed method achieves fast convergence without introducing extra per-iteration complexity.

% Fig.~\ref{fig:ee_vs_rho} shows the EE versus the maximum transmit power for different methods \Nhan{I commented many times that this writing is not clear. Instead of saying ``different methods", you can say ``the compared method", ``the conventional APG approaches/counterparts", ``the APG with backtracking and fixed step sizes", ... Many way to make it clear, why you choose the unclear way? I will not fix this kind of mistake again (as well as the one like ``two updates are computed") in your future works.}. 
Fig.~\ref{fig:ee_vs_rho} shows the EE versus the maximum transmit power for the proposed unfolded APG, the iterative APG, and the HCD schemes.
For a fair comparison, APG is executed with $T=10$ iterations, matching the number of layers in the unfolded model. The unfolded APG consistently achieves higher EE than APG and HCD across all transmit power levels for both PZF and MRT precoding. For MRT, its performance is close to that of APG under the same number of APG iterations,
% \Nhan{???}, 
indicating that APG already reaches a near-stationary performance level within 10 iterations.
The EE increases with transmit power in the low-power regime, reaches a maximum around 20–25 dBm, and then saturates or decreases. This indicates that increasing transmit power does not necessarily improve EE and highlights the importance of power control. Compared to MRT, PZF achieves significantly higher EE under all settings; for example, at $\rho^{\max}=25$ dBm, PZF provides an EE gain of approximately 198\%. In addition, the unfolded APG maintains more stable performance at higher transmit power levels, whereas APG and HCD exhibit noticeable degradation.

Fig.~\ref{fig:computation} compares the computational cost, measured in floating-point operations (FLOPs), of iterative APG and unfolded APG for different system sizes $(L,K)$. 
% For the unfolded method, the number of layers $T$ is increased with system size to maintain performance (e.g., $T=15$ for $(30,9)$ and $T=30$ for $(60,18)$).
% \Nhan{What are $T$ for other setups? You can write something like "For $L = \{20, ..., 60\}$, we set $T = \{15, 20, 25..., 30\}$, respectively. Another question: How is $T$ set for APG? or do you run it until convergence? You should mention here. Importantly, you should confirm here that the complexities of the 2 methods are compared under similar performance. Otherwise, it doesn't make sense to compare the complexity.}
For L = \{30, 40, 50, 60\}, the number of layers is set to T = \{15, 20, 25, 30\}, respectively. The iterative APG is executed until convergence, defined by a relative change in the objective value below $10^{-3}$. Notably, both methods achieve comparable EE performance under the considered settings, ensuring a fair comparison of computational complexity.
The proposed method consistently requires significantly fewer FLOPs than iterative APG across all configurations. For example, at $(L,K)=(20,6)$, iterative APG requires approximately $30$ times more FLOPs. Although $T$ increases with system size, the complexity of the unfolded method grows moderately. In contrast, larger systems lead to more challenging optimization problems, and iterative APG typically requires more iterations to converge, resulting in a faster increase in computational cost. Consequently, the complexity gap widens as $(L,K)$ increases, demonstrating the scalability of the proposed approach.

\section{Conclusion}
This paper proposed a deep-unfolded APG method for EE maximization in downlink CF-mMIMO systems under QoS and per-AP power constraints. We first developed an APG-based framework for this problem. To reduce the computational cost of the APG framework, we proposed a deep-unfolded approach by mapping the iterative updates into a finite number of trainable layers, where the step sizes and penalty parameters are learned from data, eliminating the need for line search or manual tuning. Numerical results show that the proposed method achieves EE performance comparable to iterative APG while requiring significantly lower computational cost, indicating its scalability for CF-mMIMO systems.
%=============================================================================
\appendices
\section{Proof of Theorem~\ref{thm:gradient}}
\label{app:gradient_derivation}
% =============================================================================

We derive $\nabla f_\xi = \nabla \mathrm{E_e} - \xi \nabla\Psi$ by differentiating each term with respect to $\theta_{m,t}$ of $\boldsymbol{\theta} \in \mathbb{R}^{L \times K}$.
% , where $[\nabla(\cdot)]{m,t} = \partial(\cdot)/\partial \theta{m,t}$. 
% We use the shorthand $a_{l,k}$, $b_{l,k,t}$, and $d_{l,k,t}$ as defined in~\eqref{eq:Ak}-\eqref{eq:dlkt}.

\subsection*{A.\; Gradients of $A_k^2$ and $I_k$}

Since $A_k = \sum_l a_{l,k}\theta_{l,k}$ depends only on column~$k$ of $\boldsymbol{\theta}$, the chain rule gives
\begin{equation}
    [\nabla A_k^2]_{m,t}
    = 2\,a_{m,k}\,A_k \cdot \mathbf{1}_{\{t=k\}},
    \label{eq:app_grad_DS}
\end{equation}
where $\mathbf{1}_{\{\cdot\}}$ takes the value $1$ if the condition is satisfied and $0$ otherwise.
% The interference $I_k$ %in~\eqref{eq:Ik} 
% has two variable parts: the pilot-contamination sum and the non-coherent interference; the constant~$1$ vanishes under differentiation.
Since $A_{k,t} = \sum_l b_{l,k,t}\theta_{l,t}$ for $t \in \mathcal{P}_k\setminus\{k\}$, applying the chain rule to $A_{k,t}^2$ and the power rule to $d_{l,k,t}\theta_{l,t}^2$ yields
\begin{equation}
    [\nabla I_k]_{m,t}
    = 2\,b_{m,k,t}\,A_{k,t}\cdot\mathbf{1}_{\{t\in\mathcal{P}_k\setminus\{k\}\}}
    + 2\,d_{m,k,t}\,\theta_{m,t}.
    \label{eq:app_grad_Ik}
\end{equation}

\subsection*{B.\; Gradient of the sum SE ($\nabla u$)}

The $\mathrm{SE}_k = \frac{\tau_c-\tau_p}{\tau_c} \log_2\!(1+A_k^2/I_k)$. Applying the chain rule of $\log_2(1+x)$ and the quotient rule on $A_k^2/I_k$, 
% then using $(1+A_k^2/I_k)\cdot I_k = A_k^2+I_k$:
% which separates the contributions of the signal and interference terms, 
yields
\begin{equation}
    \frac{\partial \mathrm{SE}_k}{\partial\theta_{m,t}}
    = \frac{\tau_c-\tau_p}{\tau_c\ln 2\,(A_k^2+I_k)}
    \bigl([\nabla A_k^2]_{m,t}
    - \mathrm{SINR}_k\,[\nabla I_k]_{m,t}\bigr),
    \label{eq:app_grad_Se}
\end{equation}
where $\mathrm{SINR}_k = A_k^2/I_k$. Summing~\eqref{eq:app_grad_Se} over all $k$ gives~\eqref{eq:grad_u}.

\subsection*{C.\; Gradient of the energy efficiency ($\nabla \mathrm{E_e}$)}

From~\eqref{eq:Ptotal}, $P_{\mathrm{total}} = \tilde{P} + B\bar{P}_{\mathrm{bt}}u$, where $\tilde{P} = \bar{P}_{\mathrm{fix}} + \sum_{l,k}(N_0/\alpha_l)\theta_{l,k}^2$ is independent of $u$. The power rule gives $[\nabla\tilde{P}]_{m,t} = 2(N_0/\alpha_m)\theta_{m,t} = 2N_0[\boldsymbol{\alpha}^{-1}]_{m,t}\,\theta_{m,t}$, so
\begin{equation}
    \nabla P_{\mathrm{total}}
    = 2N_0\boldsymbol{\alpha}^{-1}\!\odot\boldsymbol{\theta}
    + B\bar{P}_{\mathrm{bt}}\nabla u.
    \label{eq:app_grad_Ptotal}
\end{equation}
Applying the quotient rule to $\mathrm{E_e} = Bu/P_{\mathrm{total}}$ and substituting~\eqref{eq:app_grad_Ptotal}, then collecting $\nabla u$ terms with coefficient $P_{\mathrm{total}} - B\bar{P}_{\mathrm{bt}}u = \tilde{P}$ yields~\eqref{eq:grad_Ee}.

\subsection*{D.\; Gradient of the penalty ($\nabla\Psi$)}

From~\eqref{eq:gk}, $g_k = \sqrt{\bar{S}^{\mathrm{min}}_{k}\,I_k} - A_k$. Applying the chain rule $\partial\sqrt{I_k}/\partial\theta_{m,t} = [\nabla I_k]_{m,t}/(2\sqrt{I_k})$, substituting~\eqref{eq:app_grad_Ik}, and differentiating $-A_k$ using~\eqref{eq:app_grad_DS}, we obtain
\begin{align}
    [\nabla g_k]_{m,t}
    &= \frac{\sqrt{\bar{S}^{\mathrm{min}}_{k}}}{\sqrt{I_k}}
    \Bigl(b_{m,k,t}A_{k,t}\cdot\mathbf{1}_{\{t\in\mathcal{P}_k\setminus\{k\}\}}
    + d_{m,k,t}\theta_{m,t}\Bigr) \nonumber \\
    &\quad
    - a_{m,k}\cdot\mathbf{1}_{\{t=k\}}.
    \label{eq:app_grad_gk}
\end{align}
Since $h(x)=[\max(0,x)]^2$ has derivative $h'(x)=2[x]_+$, the chain rule $\nabla h(g_k) = h'(g_k)\nabla g_k = 2g_k\nabla g_k$ (for $g_k>0$) gives
\begin{equation}
    \nabla\Psi = \sum_{k:\,g_k>0}2\,g_k\,\nabla g_k,
    \label{eq:app_grad_Psi}
\end{equation}
which is~\eqref{eq:grad_Psi}. Substituting~\eqref{eq:grad_Ee} and~\eqref{eq:app_grad_Psi} into $\nabla f_\xi = \nabla \mathrm{E_e} - \xi\nabla\Psi$ yields the desired result.

% =============================================================================
\section*{Acknowledgment}
This work was supported in part by European Union through MSCA Doctoral Network EXACT-6G (GA 101120297).
% \ifCLASSOPTIONcaptionsoff
%   \newpage
% \fi

\bibliographystyle{IEEEtran}
\bibliography{bibtex/bib/IEEEabrv,bibtex/bib/globecom}

\end{document}